\documentclass[12pt]{article}
\usepackage{graphicx,amssymb,amsmath}

\usepackage{amsfonts}
\usepackage{booktabs}
\usepackage{cite}
\usepackage{color}
\usepackage{enumerate}
\usepackage{float}
\usepackage{hyperref}
\usepackage{mathrsfs}
\usepackage{subfig}
\usepackage{tabularx}
\usepackage{fullpage}

\newtheorem{definition}{Definition}
\newtheorem{lemma}{Lemma}
\newtheorem{theorem}{Theorem}
\newenvironment{proof}{\noindent {\textbf{Proof:}}\rm}{\hfill $\Box$
\rm\bigskip}

\title{A Linear-Time Algorithm for Discrete Radius Optimally Augmenting Paths in a Metric Space}

\author{
Haitao Wang\thanks{Department of Computer Science,
Utah State University, Logan, UT 84322, USA. {\tt haitao.wang@usu.edu}}
\and
Yiming Zhao\thanks{Corresponding author. Department of Computer Science,
Utah State University, Logan, UT 84322, USA. {\tt yiming.zhao@aggiemail.usu.edu}}
}

\begin{document}
\pagestyle{plain}
\date{}

\thispagestyle{empty}
\maketitle

\begin{abstract}

Let $P$ be a path graph of $n$ vertices embedded in a metric space. We consider the problem of adding a new edge to $P$ so that the radius of the resulting graph is minimized, where any center is constrained to be one of the vertices of $P$. Previously, the ``continuous'' version of the problem where a center may be a point in the interior of an edge of the graph was studied and a linear-time algorithm was known. Our ``discrete'' version of the problem has not been studied before. We present a linear-time algorithm for the problem.
\end{abstract}

\section{Introduction}
\label{sec:Introduction}

Let $P$ be a path graph of $n$ vertices embedded in a metric space. We wish to  add a new edge to $P$ so that the radius of the resulting graph is minimized, where
any center of the graph is constrained to be one of the vertices of $P$. The problem is formally defined as follows.

Let $\{v_1, v_2, ..., v_n\}$ be the set of vertices of $P$ along $P$. For each $i \in [1, n - 1]$, let $e(v_{i}, v_{i + 1})$ denote the edge connecting  $v_{i}$ and $v_{i + 1}$. We assume that $P$ is embedded in a metric space and $|v_iv_j|$ is the distance between two vertices $v_i$ and $v_j$, such that the following properties hold: (1) $|v_iv_j| = 0$ if and only if $i = j$; (2) $|v_iv_j| = |v_jv_i| \geq 0$; (3) $|v_iv_k| + |v_kv_j| \geq |v_iv_j|$ for any $v_k$ (i.e., the triangle inequality). Note that the length of each edge $e(v_{i}, v_{i+1})$ for $i \in [1, n - 1]$ in $P$ is equal to $|v_{i}v_{i + 1}|$. We assume that the distance $|v_iv_j|$ can be obtained in $O(1)$ time for any two vertices $v_i$ and $v_j$ of $P$.

Let $P \cup \{e(v_{i}, v_{j})\}$ denote the resulting graph (also called {\em augmenting path}) after adding a new edge $e(v_{i}, v_{j})$ connecting two vertices (i.e., $v_{i}$ and $v_{j}$) of $P$.
A vertex $c$ of $P$ is called a \emph{center} of the new graph $P \cup \{e(v_{i}, v_{j})\}$ if it minimizes the maximum length of the  shortest paths from $c$ to all vertices in the graph, and the maximum shortest path length is called the \emph{radius} of $P \cup \{e(v_{i}, v_{j})\}$. The problem is to add a new edge $e(v_{i}, v_{j})$ such that the radius of $P \cup \{e(v_{i}, v_{j})\}$ is minimized, among all vertex pairs $(v_i,v_j)$ with $1\leq i<j\leq n$. We call the problem {\em discrete radius optimally augmenting path} problem (or discrete-ROAP for short).

To the best of our knowledge, the problem has not been studied before in the literature. In this paper, we present an $O(n)$ time algorithm for the problem.

\subsection{Related work}
\label{subsec:RelatedWork}

Johnson and Wang \cite{ref:JohnsonAL19} studied a ``continuous'' version of the problem in which a center may be in the interior of an edge of the graph. In contrast, in our problem any center has to be a vertex of the graph, and thus our problem may be considered a ``discrete'' version.
Johnson and Wang \cite{ref:JohnsonAL19} gave a linear time algorithm for their continuous problem.

A similar problem that is to minimize the diameter of the augmenting path has also been studied. Gro{\ss}e et al. \cite{ref:GrobeFa15} first gave an $O(n \log^3 n)$ time algorithm; later Wang \cite{ref:WangAn18} improved the algorithm to $O(n \log n)$ time.
Variations of the diameter problem (i.e., add a new edge to $P$ to minimize the diameter of the resulting graph) were also considered. If the path $P$ is embedded in the Euclidean space $\mathbb{R}^d$ for a given constant $d$, Gro{\ss}e et al. \cite{ref:GrobeFa15} proposed an algorithm that can compute a $(1 + \epsilon)$-approximate solution for the diameter problem in $O(n + \frac{1}{{\epsilon}^3})$ time, for any $\epsilon > 0$. If $P$ is embedded in the Euclidean plane $\mathbb{R}^2$, De Carufel et al. \cite{ref:DeCarufelMi16} derived a linear-time algorithm for the continuous version of the diameter problem (i.e., the diameter is defined with respect to all points of the graph, including the points in the interior of the graph edges, not just vertices). For a geometric tree $T$ of $n$ vertices embedded in the Euclidean plane $\mathbb{R}^2$, De Carufel et al. \cite{ref:DeCarufelMi17} designed an $O(n \log n)$-time algorithm for adding a new edge to $T$ to minimize the continuous diameter in the new graph. If $T$ is a tree embedded in a metric space, Gro{\ss}e et al. \cite{ref:GrobeFa16} solved the discrete diameter problem in $O(n^2 \log n)$ time; Bil{\`o} \cite{ref:BiloAl18} improved the algorithm to $O(n \log n)$ time. Oh and Ahn \cite{ref:OhA16} considered the diameter problem on a general tree (not necessarily embedded in a metric space) and developed $O(n^2 \log^3 n)$ time algorithms for both the discrete and the continuous versions of the diameter problem; later Bil{\`o} \cite{ref:BiloAl18} gave an improved $O(n^2)$ time algorithm for the discrete diameter problem.

A more general problem is to add $k$ edges to a general graph $G$ such that the diameter of the new graph is minimized. This problem is NP-hard \cite{ref:SchooneDi97} and some variants are even W[2]-hard \cite{ref:FratiAu15, ref:GaoTh13}. Various approximation algorithms are known~\cite{ref:BiloIm12, ref:FratiAu15, ref:LiOn92}. The problem of bounding the diameters of the augmenting graphs have also been studied~\cite{ref:AlonDe00, ref:IshiiAu13}. In a geometric setting, given a circle in the plane, Bae et al. \cite{ref:BaeSh17} considered the problem of inserting $k$ shortcuts (i.e., chords) to the circle to minimize the diameter of the resulting graph.

As a motivation of our problem, we borrow an example from~\cite{ref:JohnsonAL19}. Suppose there is a highway that connects several cities and we want to build a facility along the highway to provide certain service for all these cities; it is required that the facility be located in one of the cities along the highway. In order to reduce the transportation time, one option is to construct a new highway connecting two cities such that the radius (the maximum distance from the facility to all cities) is as small as possible.

\subsection{Our approach}
\label{sec:approach}

Note that the radius of $P\cup \{e(v_i,v_j)\}$ may not be equal to its diameter divided by $2$. For example, suppose $P\cup \{e(v_i,v_j)\}$ is a cycle (i.e., $i=1$ and $j=n$) and all edges of the cycle have the same length; then one can verify that the radius of the graph is equal to its diameter.

To solve our problem, a natural idea is to see whether the algorithm~\cite{ref:JohnsonAL19} for the continuous problem can be used. To this end, two basic questions arise. First, for an augmenting graph $P \cup \{e(v_{i}, v_{j})\}$, how far a continuous center can be from the discrete center? For example, is it the case that if a continuous center lies in the interior of an edge $e$, then one of the two vertices of $e$ must be a discrete center? Second, is it the case that an optimal solution (i.e., the new edge to be added) in the continuous version must also be an optimal solution for the discrete version?

In order to answer these questions, we illustrate two examples.

Figure~\ref{fig:motivation-2} shows an example in which the path $P$ with $10$ vertices is embedded in the Euclidean plane, with $|v_{i}v_{i + 1}| = 1$ for all $1\leq i \leq 9$. Suppose a new edge $e(v_3,v_8)$ is added. It is possible to draw the figure such that $|v_{3}v_{8}| = 4$. One can verify that the only continuous center is the middle point of $e(v_{3}, v_{8})$ (whose farthest vertices are $\{v_{1}, v_{5}, v_{6}, v_{10}\}$) and the continuous radius is $4$. Either $v_{5}$ or $v_{6}$ can be a discrete center ($v_5$ has only one farthest vertex $v_{10}$ and $v_6$ has only one farthest vertex $v_1$) and the discrete radius is $5$. This example shows that the discrete center and the continuous center could be ``far from'' each other. Therefore, it is not obvious to us whether/how a continuous center can be used to find a discrete center.

\begin{figure}[t]
    \centering
    \includegraphics[scale=0.58]{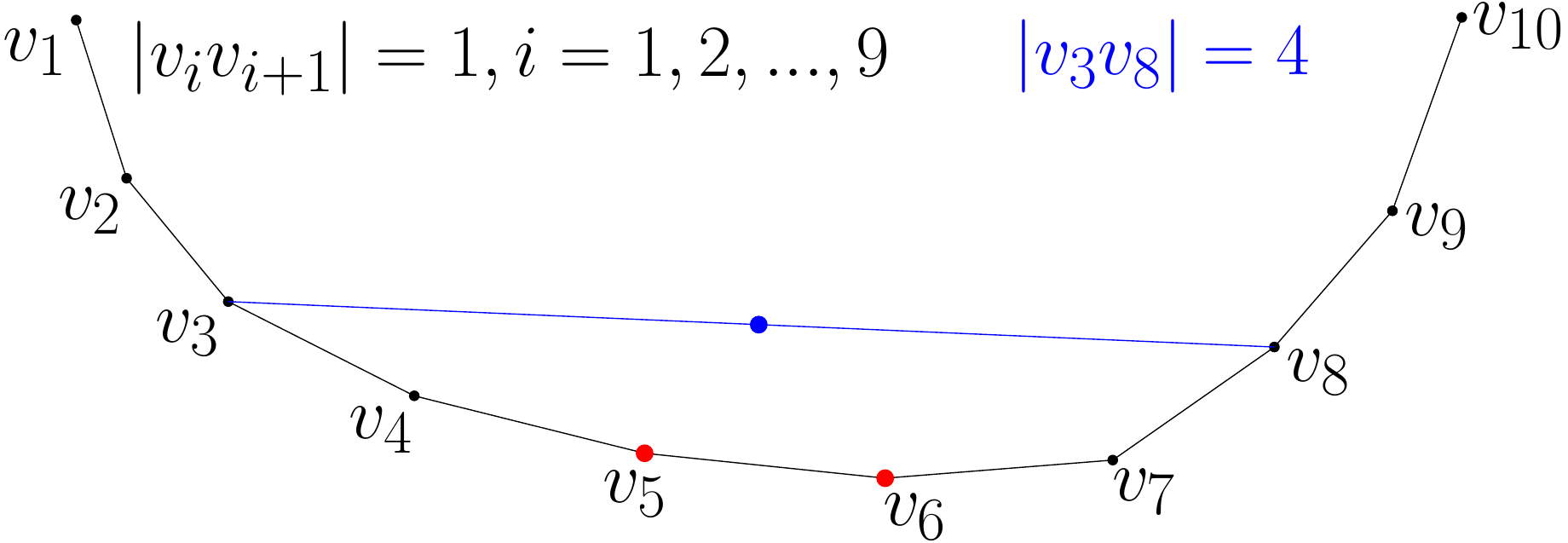}
    \caption{Illustrating the difference between the continuous center and the discrete center. The continuous center is the middle point of the new edge $e(v_{3}, v_{8})$. Either $v_{5}$ or $v_{6}$ can be a discrete center.
    }
    \label{fig:motivation-2}
\end{figure}

Figure~\ref{fig:motivation-1} shows an example in which the path $P$ with $10$ vertices is embedded in the Euclidean plane, with $|v_{i}v_{i + 1}| = 1$ for all $1\leq i\leq 9$. It is possible to draw the figure such that $|v_{3}v_{8}| = 4$, $|v_{4}v_{7}| > 2$, $|v_{5}v_{10}| > |v_{1}v_{6}|$, and $4 < |v_{1}v_{6}| < 5$. For the continuous problem, an optimal solution is to add the edge $e(v_3, v_8)$, after which the continuous center of the new graph is the middle point of $e(v_{3}, v_{8})$ (which has four farthest vertices $\{v_{1}, v_{5}, v_{6}, v_{10}\}$) and the continuous radius is $4$. For the discrete problem, an optimal solution is to add the edge $e(v_1, v_6)$, after which the discrete center of the new graph is $v_{6}$ (which has only one farthest vertex $v_{1}$) and the discrete radius is equal to $|e(v_{1}, v_{6})|$, which is larger than $4$. Note that $e(v_5,v_{10})$ is not an optimal solution due to $|v_{5}v_{10}| > |v_{1}v_{6}|$. This example shows that optimal solutions of the two versions of the problem could be very different. Therefore, it is not obvious to us whether/how a continuous optimal solution can be used to find a discrete optimal solution.

\begin{figure}[t]
    \centering
    \includegraphics[scale=0.58]{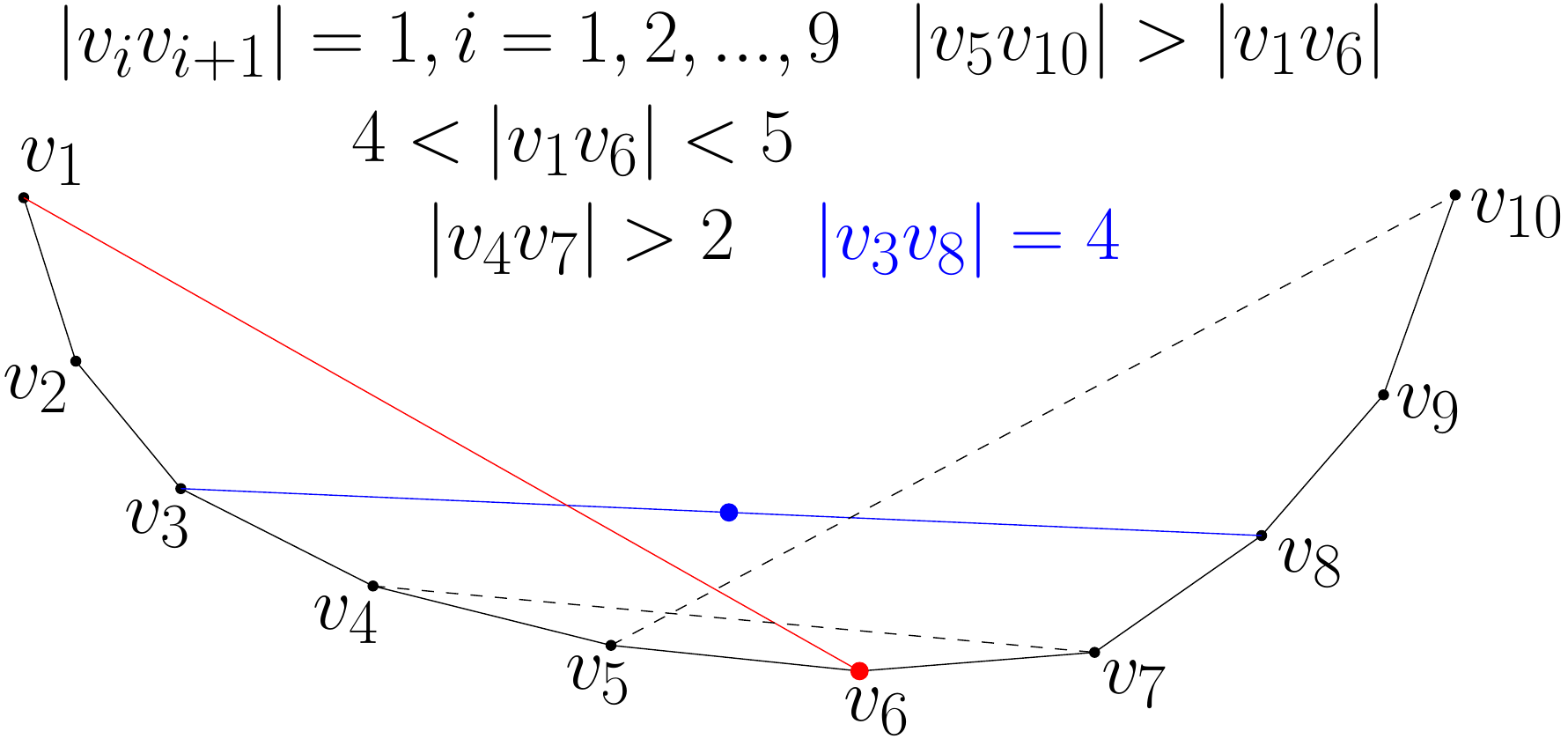}
    \caption{Illustrating the difference between an optimal solution of the continuous problem and that of the discrete problem. For the continuous problem, an optimal solution is to add the edge $e(v_3, v_8)$, after which the continuous center of the new graph is the middle point of $e(v_{3}, v_{8})$ and the continuous radius is $4$.
    For the discrete problem, an optimal solution is to add the edge $e(v_1, v_6)$, after which the discrete center of the new graph is $v_{6}$ and the discrete radius is equal to $|e(v_{1}, v_{6})|$, which is larger than $4$.}
    \label{fig:motivation-1}
\end{figure}

The above examples demonstrate that using the algorithm in~\cite{ref:JohnsonAL19} directly to solve the discrete problem seems not possible. Instead, we design a new algorithm. Our algorithm still share some similarities with that in~\cite{ref:JohnsonAL19} in the following sense. In the continuous case, any center must have two different farthest vertices in the augmenting graph. Based on the location of the center, the locations of the two farthest vertices, and the shortest paths from the center to the two farthest vertices in an optimal solution, the algorithm in~\cite{ref:JohnsonAL19} considers a constant number of configurations, and in each configuration the algorithm computes a candidate solution such that if an optimal solution conforms to the configuration, then the candidate solution is an optimal solution. In our discrete case, we also consider a constant number of configurations and process the configurations in the same way as above. However, the major difference is that the definitions of the configurations in our problem are quite different from those in~\cite{ref:JohnsonAL19}. Indeed, in our problem, a center may have only one farthest vertex.
Therefore, the configurations in our problem are defined with respect to the locations of the center and a single farthest vertex, as well as their shortest path.
In addition, unlike those in~\cite{ref:JohnsonAL19}, we do not need to consider the configurations where the center is in the interior of the new added edge. For this reason, we have much fewer configurations than those in~\cite{ref:JohnsonAL19}.

The rest of the paper is organized as follows. Section~\ref{sec:PreliminaryWork} introduces notation and definitions. Our algorithm is described in Section~\ref{sec:OurApproach}. Section~\ref{sec:conclude} concludes with some remarks.

\section{Preliminaries}
\label{sec:PreliminaryWork}

Unless otherwise stated, for any index pair $(i,j)$ or vertex pair $(v_i,v_j)$ used in our discussion, we assume that $1\leq i\leq j\leq n$.

For any two vertices $v_i$ and $v_j$ of the path $P$, we use $P(v_{i}, v_{j})$ to refer to the subpath of $P$ from $v_{i}$ to $v_{j}$ inclusively.

For any index pair $(i,j)$, define $G(i, j) = P \cup \{e(v_{i}, v_{j})\}$, i.e., the new graph after a new edge $e(v_{i}, v_{j})$ is added to the path $P$. Note that if $j=i$ or $j=i+1$, then $G(i,j)$ is $P$. Define $C({i}, {j}) = P(v_{i}, v_{j}) \cup e(v_{i}, v_{j})$, which is a cycle formed by a new edge $e(v_{i}, v_{j})$ and the subpath $P(v_{i}, v_{j})$.

For any graph $G$ used in our discussion (e.g., $G$ is $G(i,j)$, $C(i,j)$, or $P$) and any two vertices $v$ and $v'$ of $G$, we use $d_G(v,v')$ to denote the length of any shortest path from $v$ to $v'$ in $G$ and we also refer to $d_G(v,v')$ as the {\em distance} from $v$ to $v'$ in $G$. Following this definition, $d_{P}(v_{i}, v_{j})$ is the length of the subpath $P(v_i,v_j)$ and  $d_{C(i, j)}(v, v')$ is the distance between $v$ to $v'$ in the cycle $C(i, j)$. For any path $\pi$ in $G$, we use $|\pi|$ to denote the length of $\pi$. We also use $|C(i,j)|$ to denote the total length of all edges of the cycle $C(i,j)$.

Our algorithm will frequently compute $d_{P}(v_{i}, v_{j})$ for any index pair $(i,j)$. This can be done in constant time after $O(n)$ time preprocessing, e.g., compute the prefix sum $d_P(v_1,v_k)$ for all $1\leq k\leq n$.


For any vertex $v$ of any graph $G$ used in our discussion, a vertex $v'$ of $G$ is called a {\em farthest vertex} of $v$ if it maximizes $d_G(v,v')$. A vertex $v_{c}$ of $G$ is called a {\em center} if its distance to its farthest vertex is minimized, and the distance from $v_{c}$ to its farthest vertex is called the {\em radius} of $G$. Therefore, our problem is to find an index pair $(i, j)$ such that the radius of $G(i, j)$ is minimized.

Let $(i^*, j^*)$ denote an optimal solution (with $i^* < j^*$), i.e., $e(v_{i^*},v_{j^*})$ is the new edge to be added. Let $c^*$ denote the index of a center of $G(i^*, j^*)$, $r^*$ the radius of $G(i^*,j^*)$, $a^*$ the index of a farthest vertex of $v_{c^*}$ in $G(i^*,j^*)$, and $\pi^*$ a shortest path from $v_{c^*}$ to $v_{a^*}$ in $G(i^*,j^*)$. Note that the center of $G(i^*, j^*)$ may not be unique, in which case we use $c^*$ to refer to an arbitrary one, but once $c^*$ is fixed we will never change it throughout the paper. So as $a^*$ and $\pi^*$. Note that $c^*\neq a^*$ since otherwise the graph would have only one vertex.






\section{The Algorithm}
\label{sec:OurApproach}


As discussed in Section~\ref{sec:approach}, we consider a constant number of configurations for the optimal solution $G(i^*,j^*)$. For each configuration, we compute in $O(n)$ time a candidate solution (consisting of an index pair $(i',j')$, a candidate center $c'$ and a candidate radius $r'$) such that if the optimal solution conforms to the configuration, then our candidate solution is an optimal one, i.e., $r^*=r'$. On the other hand, the candidate solution is a {\em feasible} one, i.e., the distances from $c'$ to all vertices in $G(i',j')$ are at most $r'$.

In the following, we first give an overview of all configurations and then present algorithms to compute candidate solutions for them.

\subsection{Configuration overview}
The configurations are defined with respect to the locations of  $v_{a^*}$ and $v_{c^*}$ as well as whether the path $\pi^*$ contains the new edge $e(v_{i^*},v_{j^*})$.

Depending on whether $c^*\in (i^*,j^*)$, there are two main cases.

\begin{description}
\item[Case 1: $c^*\not\in (i^*,j^*)$.]

In this case, $c^*$ is either in $[1,i^*]$ or in $[j^*,n]$. Hence, there are two subcases.

\begin{description}
\item[Case 1.1:] $c^*\in [1,i^*]$. See Fig.~\ref{fig:case1-1}.

\item[Case 1.2:] $c^*\in [j^*,n]$.

This case is symmetric to Case 1.1.
\end{description}

\item[Case 2: $c^*\in (i^*,j^*)$.]

Notice that $a^*$ cannot be in $[2,i^*]\cup [j^*,n-1]$. Hence, there are three subcases $a^*=1$, $a^*=n$, and $a^*\in (i^*,j^*)$.

\begin{description}
\item[Case 2.1:] $a^*=1$.

This case further has two subcases depending on whether the new edge $e(v_i^*,v_j^*)$ is contained in the path $\pi^*$.

\begin{description}
\item[Case 2.1.1:] $e(v_i^*,v_j^*)\subseteq \pi^*$. See Fig.~\ref{fig:case2-1-3-1}.
\item[Case 2.1.2:] $e(v_i^*,v_j^*)\not\subseteq \pi^*$. See Fig.~\ref{fig:case2-1-3-2}.
\end{description}

\item[Case 2.2:] $a^*=n$. This case is symmetric to Case~2.1.

\item[Case 2.3:] $a^*\in (i^*,j^*)$.




\end{description}
\end{description}

In fact, we will only compute candidate solutions for Case 1.1 and Case 1.2. We will show that other cases can be reduced to these two cases (i.e., if any case other than Case 1.1 and Case 1.2 has an optimal solution, then one of Case 1.1 and Case 1.2 must have an optimal solution).

\subsection{Computing candidate solutions}
We are now in a position to describe our algorithms for computing candidate solutions.

\subsubsection*{Case 1: $c^*\not\in(i^*,j^*)$.}
\label{subsubsec:case1}

Depending on whether $c^* \in [1,i^*]$ or $c^* \in [j^*,n]$, there are two subcases.

\subsubsection*{Case 1.1: $c^* \in [1, i^*]$.}
\label{subsecsec:case1-1}

\begin{figure}[t]
    \centering
    \includegraphics[scale=0.55]{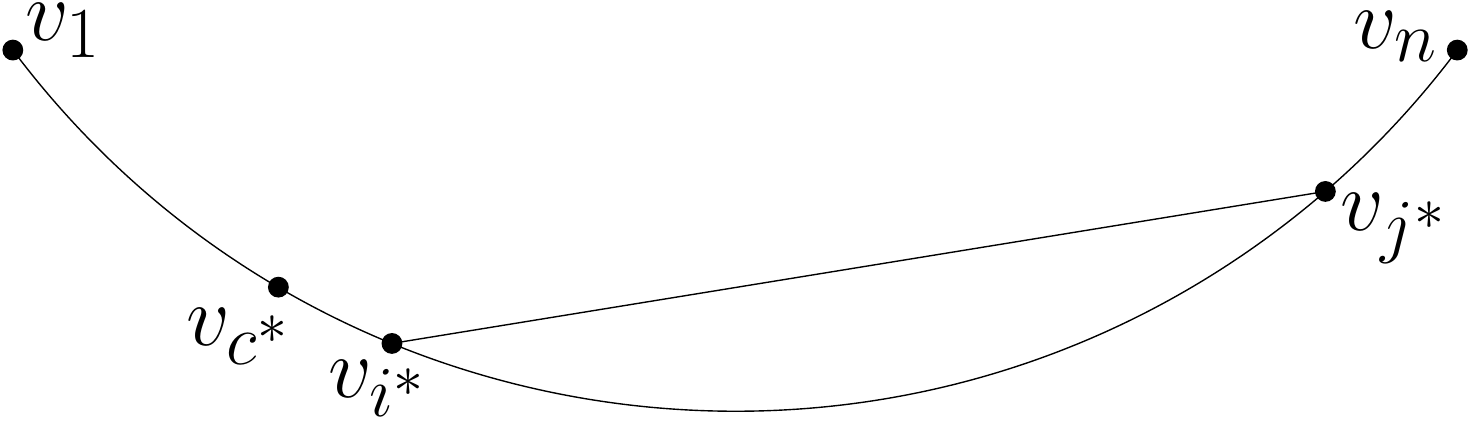}
    \caption{Illustrating the configuration for Case 1.1, where $c^* \in [1, i^*]$.}
    \label{fig:case1-1}
\end{figure}

Refer to Fig.~\ref{fig:case1-1}. In this case, it is not difficult to see
that either $a^*=1$ or $a^*\in [i^*,n]$ and thus the radius $r^*$ is equal to $$\max \{d_{P}(v_{1}, v_{c^*}), d_{P}(v_{c^*}, v_{i^*}) + \max_{k \in [i^*, n]} d_{G({i^*}, {j^*})} (v_{i^*}, v_{k})\}.$$


\begin{definition}
For each $i\in [1,n-1]$, define
    \begin{gather*}
        \lambda_{i} = \min_{j \in [i, n]} \max_{k \in [i, n]} d_{G(i, j)} (v_{i}, v_{k}), \\
        j_i = \arg \min_{j \in [i, n]} \max_{k \in [i, n]} d_{G(i, j)} (v_{i}, v_{k}),\\
        r_i = \min_{k\in [1,i]}\max\{d_P(v_1,v_k),d_P(v_k,v_i)+\lambda_i\},\\
        c_i = \arg \min_{k\in [1,i]}\max\{d_P(v_1,v_k),d_P(v_k,v_i)+\lambda_i\}.
    \end{gather*}
\end{definition}

The values $\lambda_i$ and $j_i$ were also used for solving the continuous problem in~\cite{ref:JohnsonAL19}, where an algorithm was given that can compute $\lambda_{i}$ and $j_i$ for all $i=1,2,\ldots, n-1$ in $O(n)$ time. For our discrete problem, we also need to compute $r_i$ and $c_i$ for all $i=1,2,\ldots, n-1$. To this end, we propose an $O(n)$-time algorithm in Lemma~\ref{lem:10}.

\begin{lemma}\label{lem:10}
The values $r_i$ and $c_i$ for all $i=1,2,\ldots, n-1$ can be computed in $O(n)$ time.
\end{lemma}
\begin{proof}
We first compute $\lambda_{i}$ for all $i=1,2,\ldots, n-1$ in $O(n)$ time~\cite{ref:JohnsonAL19}. Note that once $c_i$ for all $i=1,2,\ldots, n-1$ are known, all $r_i$ can be computed in additional $O(n)$ time because $r_i=\max\{d_P(v_1,v_{c_i}),d_P(v_{c_i},v_i)+\lambda_i)\}$. Hence, we will focus on computing $c_i$ below.

For each $i\in [1,n-1]$, define $k_i$ as the largest index $k\in [1,i]$ such that $d_P(v_1,v_k)\leq d_P(v_k,v_i)+\lambda_i$.

We claim that for each $i\in [1,n-1]$, $c_i$ is either $k_i$ or $k_i+1$. Indeed, as $k$ changes $[1,i]$, the value $d_P(v_1,v_k)$ is monotonically increasing while the value $d_P(v_k,v_i)+\lambda_i$ is monotonically decreasing. By the definition of $c_i$ and $k_i$, the claim follows.

In light of the claim, once $k_i$ is known, $c_i$ can be determined in additional $O(1)$ time. In the following, we describe an algorithm to compute $k_i$ for all $i=1,2,\ldots,n-1$ in $O(n)$ time.

We first prove a critical monotonicity property: $k_i\leq k_{i+1}$ for all $i\in [1,n-2]$. To this end, it suffices to show that $d_P(v_1,v_{k_i})\leq d_P(v_{k_i},v_{i+1})+\lambda_{i+1}$. We claim that $\lambda_i\leq d_P(v_i,v_{i+1})+\lambda_{i+1}$. Before proving the claim, we use the claim to prove the monotonicity property:
  \begin{equation*}
  \begin{split}
          d_P(v_1,v_{k_i}) & \leq d_P(v_{k_i},v_i) + \lambda_i\\
        & = d_P(v_{k_i},v_{i+1}) - d_P(v_i,v_{i+1}) +\lambda_i\\
        & \leq d_P(v_{k_i},v_{i+1}) + \lambda_{i+1}.
    \end{split}
    \end{equation*}
The first inequaltiy is due to the definition of $k_i$ while the last inequality is due to the above claim. This proves the monotonicity property. In the following we prove the claim.
The proof involves two graphs $G(i, j_{i+1})$ and $G(i+1, j_{i+1})$, e.g., see Fig.~\ref{fig:twographs}.

\begin{figure}[h]
    \centering
    \includegraphics[scale=0.55]{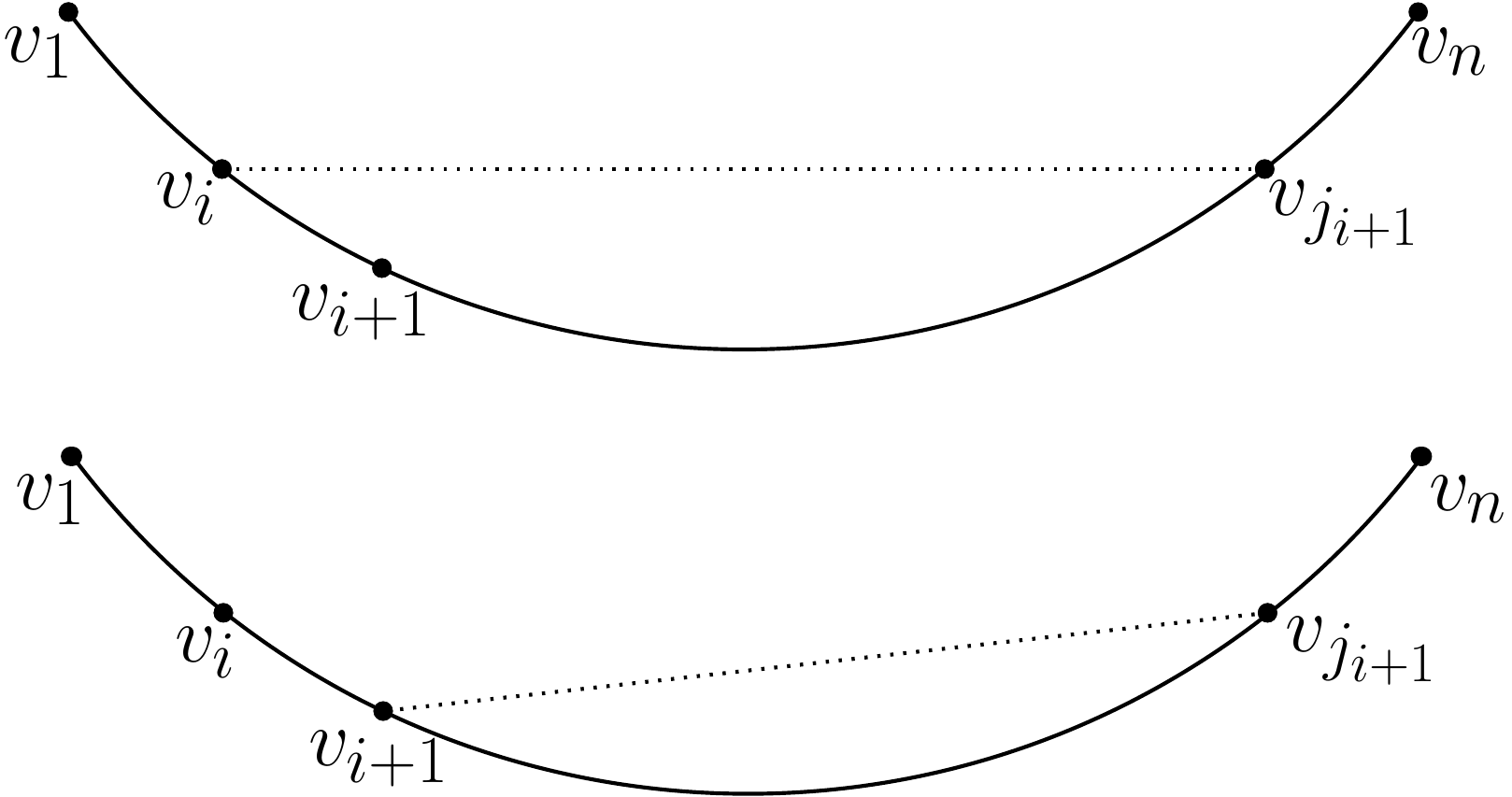}
    \caption{Illustrating the two graphs $G(i, j_{i+1})$ (top) and $G(i+1, j_{i+1})$ (bottom).}
    \label{fig:twographs}
\end{figure}

By definition, $\lambda_i\leq \max_{k \in [i, n]} d_{G(i, j_{i+1})} (v_{i}, v_{k})$. Define $k'=\arg \max_{k \in [i, n]} d_{G(i, j_{i+1})} (v_{i}, v_{k})$. Hence, $\lambda_i\leq d_{G(i, j_{i+1})} (v_{i}, v_{k'})$. It is not difficult to see that either $k'=n$ or $k'\in (i,j_{i+1})$. Below we prove $\lambda_i\leq d_P(v_i,v_{i+1})+\lambda_{i+1}$ for each case.

If $k'=n$, then due to the triangle inequality, $d_{G(i, j_{i+1})} (v_{i}, v_{k'}) = |v_iv_{j_{i+1}}|+d_P(v_{j_{i+1}},v_n)$.
Also due to the triangle inequality, $|v_iv_{j_{i+1}}|\leq d_P(v_i,v_{i+1})+|v_{i+1}v_{j_{i+1}}|$ and $d_{G(i+1, j_{i+1})} (v_{i+1}, v_{n}) = |v_{i+1}v_{j_{i+1}}|+d_P(v_{j_{i+1}},v_n)$. In addition, according to the definition of $\lambda_{i+1}$, we have $\lambda_{i+1}=\max_{k \in [i+1, n]} d_{G(i+1, j_{i+1})} (v_{i+1}, v_{k})\geq d_{G(i+1, j_{i+1})} (v_{i+1}, v_{n})$. Combining all above, we can derive
  \begin{equation*}
  \begin{split}
         \lambda_i &\leq d_{G(i, j_{i+1})} (v_{i}, v_{k'}) =          d_{G(i, j_{i+1})} (v_{i}, v_{n}) \\
                & = |v_iv_{j_{i+1}}|+d_P(v_{j_{i+1}},v_n)\\
                &\leq d_P(v_i,v_{i+1})+|v_{i+1}v_{j_{i+1}}|   +d_P(v_{j_{i+1}},v_n)\\
                & = d_P(v_i,v_{i+1})+ d_{G(i+1, j_{i+1})} (v_{i+1}, v_{n})\\
                &\leq d_P(v_i,v_{i+1}) + \max_{k \in [i+1, n]} d_{G(i+1, j_{i+1})} (v_{i+1}, v_{k})\\
                & = d_P(v_i,v_{i+1}) + \lambda_{i+1}.
    \end{split}
    \end{equation*}

We proceed to the case $k'\in (i,j_{i+1})$. Consider the graph $G(i+1,j_{i+1})$. $d_{G(i+1, j_{i+1})} (v_{i+1}, v_{k'})$ is equal to either $d_{P}(v_{i+1},v_{k'})$ or $|v_{i+1}v_{j_{i+1}}|+d_P(v_{k'},v_{j_{i+1}})$.

In the former case, we have
 \begin{equation*}
  \begin{split}
         \lambda_i &\leq d_{G(i, j_{i+1})} (v_{i}, v_{k'}) \leq d_{P} (v_{i}, v_{k'}) \\
                & = d_{P} (v_{i}, v_{i+1}) + d_{P} (v_{i+1}, v_{k'})\\
                & = d_{P} (v_{i}, v_{i+1}) + d_{G(i+1, j_{i+1})} (v_{i+1}, v_{k'})\\
                &\leq d_P(v_i,v_{i+1}) + \max_{k \in [i+1, n]} d_{G(i+1, j_{i+1})} (v_{i+1}, v_{k})\\
                & = d_P(v_i,v_{i+1}) + \lambda_{i+1}.
    \end{split}
    \end{equation*}

In the latter case, similarly we can derive
     \begin{equation*}
  \begin{split}
         \lambda_i &\leq d_{G(i, j_{i+1})} (v_{i}, v_{k'}) \leq                |v_{i}v_{j_{i+1}}|+d_P(v_{k'},v_{j_{i+1}})\\
                &\leq d_P(v_i,v_{i+1})+|v_{i+1}v_{j_{i+1}}|+d_P(v_{k'},v_{j_{i+1}})\\
                & = d_P(v_i,v_{i+1})+ d_{G(i+1, j_{i+1})} (v_{i+1}, v_{k'})\\
                &\leq d_P(v_i,v_{i+1}) + \max_{k \in [i+1, n]} d_{G(i+1, j_{i+1})} (v_{i+1}, v_{k})\\
                & = d_P(v_i,v_{i+1}) + \lambda_{i+1}.
    \end{split}
    \end{equation*}
This proves the claim and thus the monotonicity property of $k_i$'s.

Using the monotonicity property of $k_i$'s, we can easily compute all $k_i$'s in $O(n)$ time as follows. Starting from $i=1$, the algorithm incrementally computes $k_i$ for all $i=1,2,\ldots, n-1$. The algorithm maintains an index $k$. Initially, $k=i=1$ and $k_i=k$. Consider a general step where $k_i$ has just been computed and $k=k_i$. Next we compute $k_{i+1}$ as follows. As long as  $d_P(v_1,v_{k+1})\leq d_P(v_{k+1},v_{i+1})+\lambda_{i+1}$, we increment $k$ by one. After that, we set $k_{i+1}=k$. The  monotonicity property of $k_i$'s guarantees the correctness of the algorithm. The running time is $O(n)$.

The lemma is thus proved.
\end{proof}

We obtain a candidate solution for this configuration as follows.
We first compute $\lambda_{i}$ and $j_i$ for all $i=1,2,\ldots, n-1$ in $O(n)$ time~\cite{ref:JohnsonAL19}. We then use Lemma~\ref{lem:10} to compute $r_i$ and $c_i$ for all $i=1,2,\ldots,n-1$. Let $i'=\arg\min_{i\in [1,n-1]}r_i$. Let $r'=r_{i'}$ and $j'=j_{i'}$.
We return $(i',j')$, $c'$, and $r'$ as a candidate solution for this configuration. Notice that the candidate solution is a feasible solution, i.e., the distances from $v_{c'}$ to all vertices in $G(v_{i'},v_{j'})$ are  at most $r'$. The following lemma establishes the correctness of our candidate solution.


\begin{lemma}
$r'=r^*$.
\end{lemma}
\begin{proof}
First of all, as the candidate solution is a feasible one, by the definition of $r^*$, $r^*\leq r'$ holds. It remains to prove $r'\leq r^*$.

Recall that $r^* = \max \{d_{P}(v_{c^*}, v_{1}), d_{P}(v_{c^*}, v_{i^*}) + \max_{k \in [i^*, n]} d_{G({i^*}, {j^*})} (v_{i^*}, v_{k})\}$. By the definition of $\lambda_i$, $\lambda_{i^*} \leq \max_{k \in [i^*, n]} d_{G({i^*}, {j^*})} (v_{i^*}, v_{k})$. Thus, $r^* \geq \max \{d_{P}(v_{c^*}, v_{1}), d_{P}(v_{c^*}, v_{i^*}) + \lambda_{i^*}\}$. We claim that $r^* = \max \{d_{P}(v_{c^*}, v_{1}), d_{P}(v_{c^*}, v_{i^*}) + \lambda_{i^*}\}$. Indeed, we know that the value $\max \{d_{P}(v_{c^*}, v_{1}), d_{P}(v_{c^*}, v_{i^*}) + \lambda_{i^*}\}$ is equal to the distance from vertex $c^*$ to its farthest vertex in the graph $G({i^*},{j_{i^*}})$. By the definition of $r^*$, $r^*\leq \max \{d_{P}(v_{c^*}, v_{1}), d_{P}(v_{c^*}, v_{i^*}) + \lambda_{i^*}\}$. The claim thus follows.

The claim and the definition of $r_{i^*}$ together lead to $r_{i^*}\leq r^*$. Further, by the definition of the index $i'$, we have $r'=r_{i'}\leq r_{i^*}\leq r^*$. The lemma thus follows.
\end{proof}


\subsubsection*{Case 1.2: $c^* \in [j^*, n]$.}
\label{subsubsec:case1-2}

This case is symmetric to Case 1.1 and we use a similar algorithm to compute a candidate solution. The details are omitted.

\subsubsection*{Case 2: $c^* \in (i^*, j^*)$.}
\label{subsubsec:case2}

We now consider the case $c^* \in (i^*, j^*)$.
In this case, it is easy to see that $d_{G(i^*,j^*)}(v_{c^*},v_k)<d_{G(i^*,j^*)}(v_{c^*},v_1)$ for any $k\in (1,i^*]$ and similarly $d_{G(i^*,j^*)}(v_{c^*},v_k)<d_{G(i^*,j^*)}(v_{c^*},v_n)$ for any $k\in [j^*,n)$. Hence, $a^*$ cannot be in $(1,i^*]\cup [j^*,n)$. Thus, $a^* = 1$, $a^* = n$, or $a^* \in (i^*, j^*)$.

\subsubsection*{Case 2.1: $a^* = 1$.}
\label{subsubsec:case2-1}

Depending on whether the new added edge $e(v_{i^*}, v_{j^*})$ is contained in the path $\pi^*$, there are two cases.

\subsubsection*{Case 2.1.1: $e(v_{i^*}, v_{j^*}) \subseteq \pi^*$.}
\label{subsubsec:case2-1-3-1}

\begin{figure}[t]
    \centering
    \includegraphics[scale=0.55]{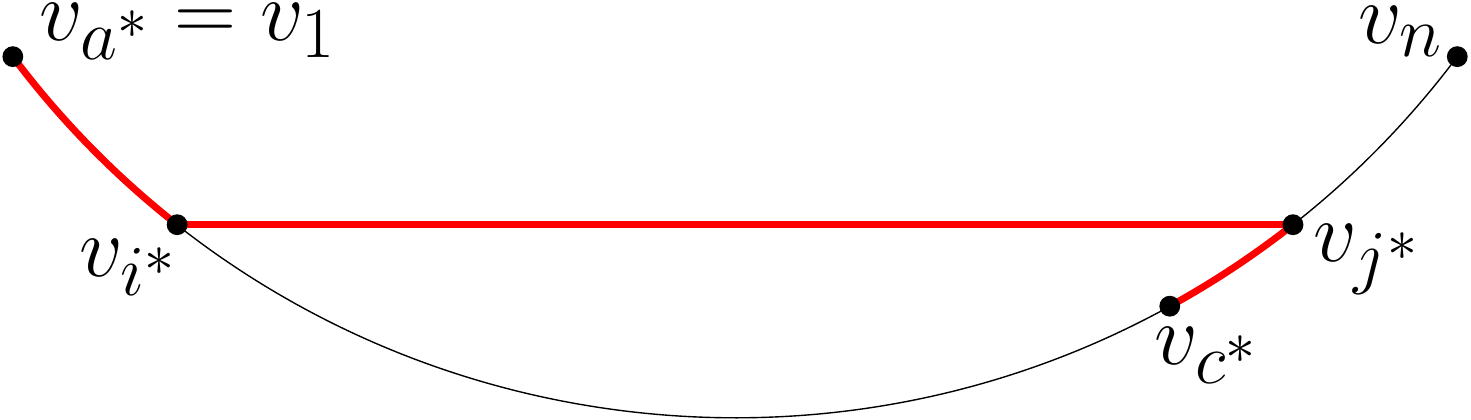}
    \caption{Illustrating the configuration for Case 2.1.1, where $c^* \in ({i^*}, {j^*})$, $a^* = 1$, and $e(v_{i^*}, v_{j^*}) \subseteq \pi^*$. The thick (red) path is $\pi^*$.}
    \label{fig:case2-1-3-1}
\end{figure}

In this case, $c^* \in ({i^*}, {j^*})$, $a^* = 1$, and $e(v_{i^*}, v_{j^*}) \subseteq \pi^*$. This implies that $\pi^*=P(v_{c^*},v_{j^*})\cup e(v_{i^*},v_{j^*})\cup P(v_1,v_{i^*})$; e.g., see Fig.~\ref{fig:case2-1-3-1}.

\begin{lemma}\label{lem:40}
The index pair $(i^*,c^*)$ is an optimal solution and $c^*$ is a center of the graph $G({i^*},{c^*})$.
\end{lemma}
\begin{proof}
We show that the distances from $c^*$ to all vertices in the graph $G({i^*},{c^*})$ are at most $r^*$ (e.g., see Fig.~\ref{fig:case2-1-3-1-proof}). This implies that the radius of $G({i^*},{c^*})$ is at most $r^*$ and thus proves the lemma.

\begin{figure}[h]
    \centering
    \includegraphics[scale=0.55]{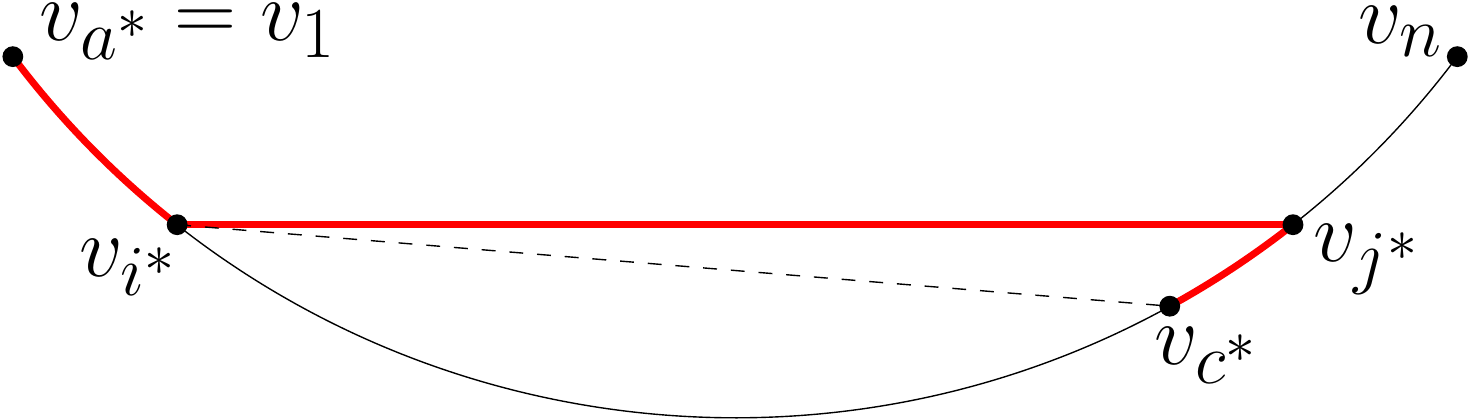}
    \caption{Illustrating the proof of Lemma~\ref{lem:40}: The distances from $c^*$ to all vertices in $G({i^*}, c^*)$ are at most $r^*$.}
    \label{fig:case2-1-3-1-proof}
\end{figure}

Let $k$ be any index in $[1,n]$. Our goal is to prove $d_{G({i^*},{c^*})}(v_{c^*},v_k)\leq r^*$.

If $k\in [1,i^*]$, then $d_{G({i^*},{c^*})}(v_{c^*},v_k)\leq |e(v_{i^*},v_{c^*})|+d_P(v_k,v_{i^*})$. By the triangle inequality, $|e(v_{i^*},v_{c^*})|\leq |e(v_{i^*},v_{j^*})|+d_P(v_{c^*},v_{j^*})$ holds. Hence, $r^*=|\pi^*|=|e(v_{i^*},v_{j^*})|+d_P(v_{c^*},v_{j^*})+d_P(v_1,v_{i^*})\geq |e(v_{i^*},v_{c^*})|+d_P(v_1,v_{i^*})\geq |e(v_{i^*},v_{c^*})|+d_P(v_k,v_{i^*})\geq d_{G({i^*},{c^*})}(v_{c^*},v_k)$.

If $k\in [c^*,n]$, then $d_{G({i^*},{c^*})}(v_{c^*},v_k)\leq d_P(v_{c^*},v_k)\leq d_P(v_{c^*},v_n)$. As $\pi^*$ is a shortest path in $G({i^*},{j^*})$ and $\pi^*$ contains $P(v_{c^*},v_{j^*})$, $P(v_{c^*},v_{n})$ must be a shortest path from $v_{c^*}$ to $v_n$ in $G({i^*},{j^*})$, implying that $d_P(v_{c^*},v_n)\leq r^*$. Therefore, $d_{G({i^*},{c^*})}(v_{c^*},v_k)\leq r^*$ holds.

If $k\in (i^*,c^*)$, then both $v_{c^*}$ and $v_k$ are in the cycle $C({i^*},{j^*})$ of the graph $G({i^*},{j^*})$ and are also in the cycle $C({i^*},{c^*})$ of the graph $G({i^*},{c^*})$. Hence, $d_{G({i^*},{j^*})}(v_{c^*},v_k)=d_{C({i^*},{j^*})}(v_{c^*},v_k)$ and  $d_{G({i^*},{c^*})}(v_{c^*},v_k)=d_{C({i^*},{c^*})}(v_{c^*},v_k)$. Due to the triangle inequality, $|C({i^*},{c^*})|\leq |C({i^*},{j^*})|$. Hence,  $d_{C({i^*},{c^*})}(v_{c^*},v_k)\leq d_{C({i^*},{j^*})}(v_{c^*},v_k)$. As  $d_{G({i^*},{j^*})}(v_{c^*},v_k)\leq r^*$, we can now obtain $d_{G({i^*},{c^*})}(v_{c^*},v_k)=d_{C({i^*},{c^*})}(v_{c^*},v_k)\leq d_{C({i^*},{j^*})}(v_{c^*},v_k) = d_{G({i^*},{j^*})}(v_{c^*},v_k)\leq r^*$.
\end{proof}

Because $(i^*,c^*)$ is an optimal solution with $c^*$ as a center in the graph $G(i^*,c^*)$, it is a configuration of Case~1.2. Hence, the candidate solution found by our algorithm for Case~1.2 is also an optimal solution. Therefore, it is not necessary to compute a candidate solution for this case any more. In other words, this case can be reduced to Case~1.2.


\subsubsection*{Case 2.1.2: $e(v_{i^*}, v_{j^*}) \not\subseteq \pi^*$.}
\label{subsubsec:case2-1-3-2}

\begin{figure}[htbp]
    \centering
    \includegraphics[scale=0.55]{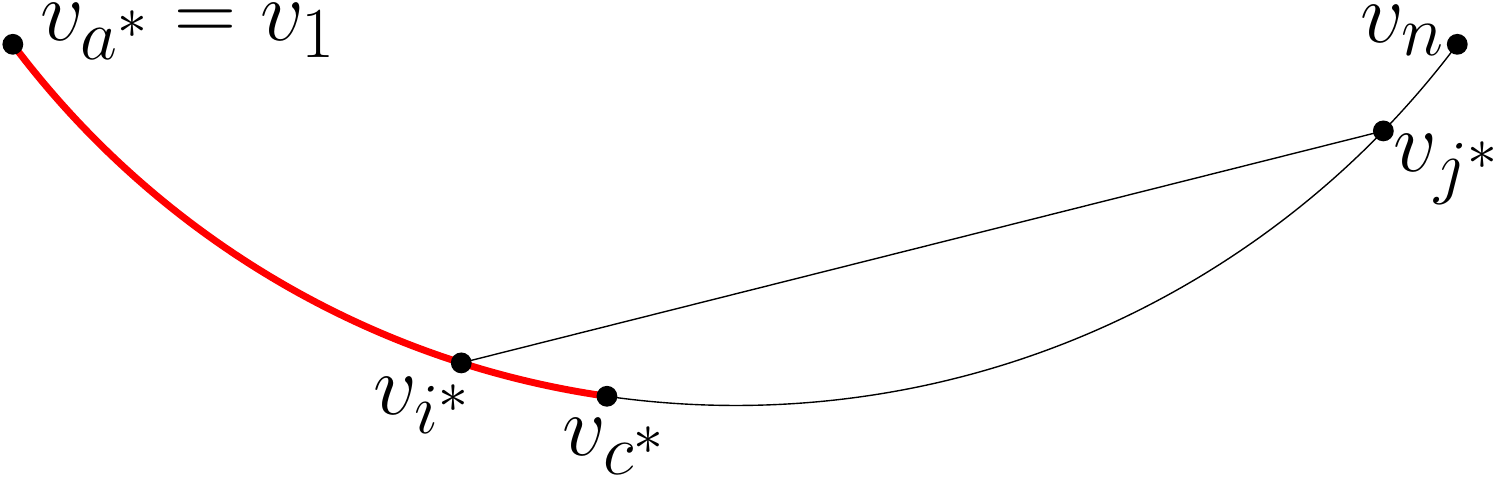}
    \caption{Illustrating the configuration for Case 2.1.2, where $c^* \in ({i^*}, {j^*})$, $a^* = 1$, and $e(v_{i^*}, v_{j^*}) \not\subseteq \pi^*$. The thick (red) path is $\pi^*$.}
    \label{fig:case2-1-3-2}
\end{figure}

Refer to Fig. \ref{fig:case2-1-3-2}. In this case, $c^* \in ({i^*}, {j^*})$, $a^* = 1$, and $e(v_{i^*}, v_{j^*}) \not\subseteq \pi^*$.
This implies that $\pi^*=P(v_1,v_{c^*})$.
The following lemma reduces this case to Case 1.1.

\begin{lemma}\label{lem:50}
The index pair $(c^*,j^*)$ is an optimal solution and $c^*$ is a center of the graph $G(c^*,j^*)$.
\end{lemma}
\begin{proof}
Some proof techniques are similar to those for Lemma~\ref{lem:40}.
It suffices to show that the distances from $c^*$ to all vertices in $G({c^*},{j^*})$ are at most $r^*$ (e.g., see Fig.~\ref{fig:case2-1-3-2-proof}).
Let $k$ be any index in $[1,n]$. Our goal is to prove $d_{G({c^*},{j^*})}(v_{c^*},v_k)\leq r^*$.

\begin{figure}[h]
    \centering
    \includegraphics[scale=0.55]{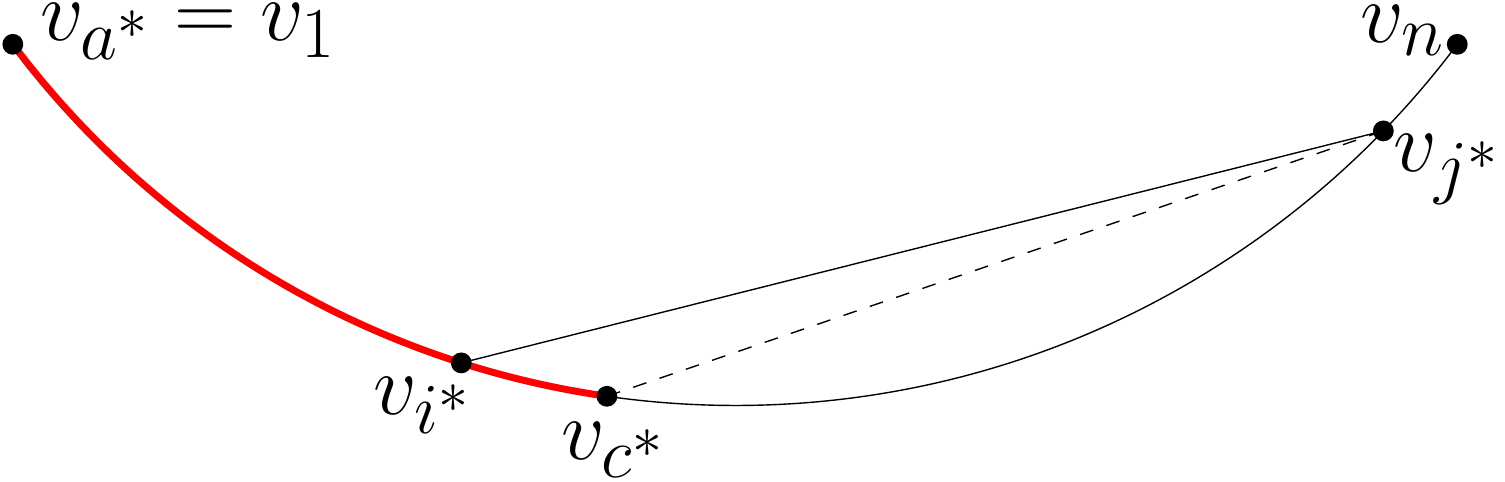}
    \caption{Illustrating the proof of Lemma~\ref{lem:50}: The distances from $c^*$ to all vertices in $G(c^*, j^*)$ are at most $r^*$.}
    \label{fig:case2-1-3-2-proof}
\end{figure}

Note that $d_P(v_{c^*},v_1)=r^*$, for $\pi^*=P(v_1,v_{c^*})$.

If $k\in [1,c^*]$, then $d_{G({c^*},{j^*})}(v_{c^*},v_k)\leq d_P(v_{c^*},v_k)\leq d_P(v_{c^*},v_1)=r^*$.

If $k\in [j^*,n]$, then $d_{G({c^*},{j^*})}(v_{c^*},v_k)\leq d_{G({c^*},{j^*})}(v_{c^*},v_n)$. Below we prove $d_{G({c^*},{j^*})}(v_{c^*},v_n)\leq d_{G({i^*},{j^*})}(v_{c^*},v_n)$, which is at most $r^*$. Note that $d_{G({i^*},{j^*})}(v_{c^*},v_n)=\min\{d_P(v_{c^*},v_n), d_P(v_{c^*},v_{i^*})+|e(v_{i^*},v_{j^*})|+d_P(v_{j^*},v_n)\}$. If $d_{G({i^*},{j^*})}(v_{c^*},v_n)=d_P(v_{c^*},v_n)$, then we have $d_{G({c^*},{j^*})}(v_{c^*},v_n)\leq d_P(v_{c^*},v_n)=d_{G({i^*},{j^*})}(v_{c^*},v_n)$. If $d_{G({i^*},{j^*})}(v_{c^*},v_n)=d_P(v_{c^*},v_{i^*})+|e(v_{i^*},v_{j^*})|+d_P(v_{j^*},v_n)$, then by the triangle inequality, $d_{G({c^*},{j^*})}(v_{c^*},v_n)\leq |e(v_{c^*},v_{j^*})|+d_P(v_{j^*},v_n)\leq d_P(v_{c^*},v_{i^*})+|e(v_{i^*},v_{j^*})|+d_P(v_{j^*},v_n)=d_{G({i^*},{j^*})}(v_{c^*},v_n)$.

If $k\in (c^*,j^*)$, then both $v_{c^*}$ and $v_k$ are in the cycle $C({i^*},{j^*})$ of the graph $G({i^*},{j^*})$ and are also in the cycle $C({c^*},{j^*})$ of the graph $G({c^*},{j^*})$. By a similar analysis as that for Lemma~\ref{lem:40}, we can obtain $d_{G({c^*},{j^*})}(v_{c^*},v_k)\leq d_{G({i^*},{j^*})}(v_{c^*},v_k)\leq r^*$.
\end{proof}


\subsubsection*{Case 2.2: $a^* = n$.}
\label{subsubsec:case2-2}

This case is symmetric to Case 2.1 and we omit the details.

\subsubsection*{Case 2.3: $a^* \in (i^*, j^*)$.}
\label{subsubsec:case2-3}

In this case, both $a^*$ and $c^*$ are in $(i^*,j^*)$. Without loss of generality, we assume that $c^*<a^*$. The following lemma reduces this case to Case 1.1.

\begin{lemma}\label{lem:60}
The index pair $(c^*,j^*)$ is an optimal solution and $c^*$ is a center of the graph $G(c^*,j^*)$.
\end{lemma}
\begin{proof}
It suffices to show that the distances from $c^*$ to all vertices in $G({c^*},{j^*})$ are at most $r^*$ (e.g., see Fig.~\ref{fig:case2-3-1-2}).
Let $k$ be any index in $[1,n]$. Our goal is to prove $d_{G({c^*},{j^*})}(v_{c^*},v_k)\leq r^*$.

\begin{figure}[htbp]
    \centering
    \includegraphics[scale=0.55]{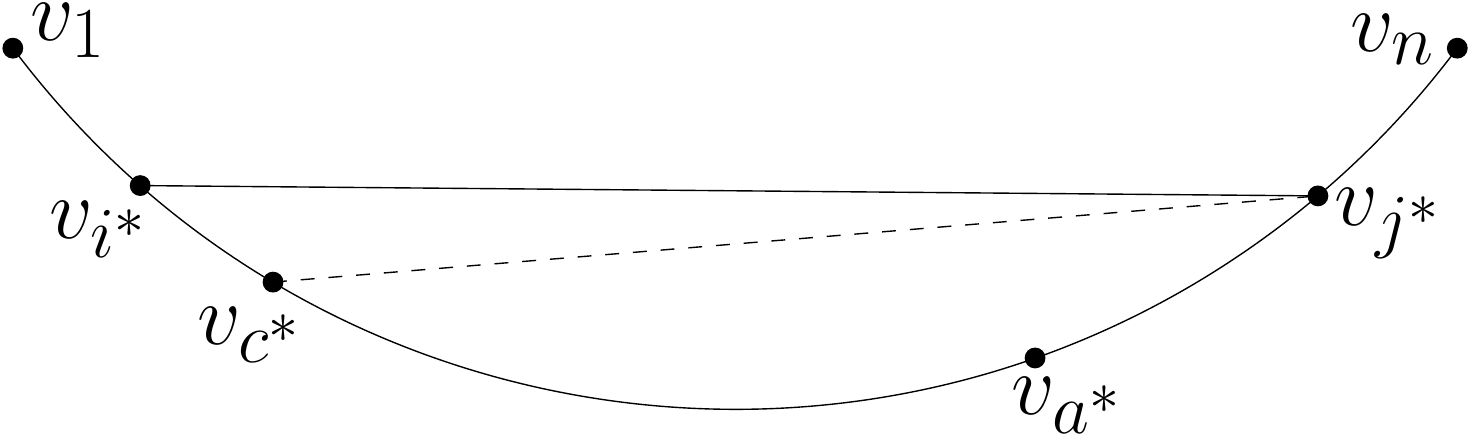}
    \caption{Illustrating the proof of Lemma~\ref{lem:60}: The distances from $c^*$ to all vertices in $G(c^*, j^*)$ are at most $r^*$.}
    \label{fig:case2-3-1-2}
\end{figure}

We first prove a claim that $P(v_{c^*},v_{j^*})$ is not a
shortest path from $v_{c^*}$ to $v_{j^*}$ in
$G({i^*},{j^*})$. Assume to the contrary this is not
true. Then, as $v_{a^*}\in P(v_{c^*},v_{j^*})$,
$P(v_{c^*},v_{a^*})$ must be a shortest path from $v_{c^*}$ to
$v_{a^*}$ and thus its length is equal to $r^*$. Because $a^*\in (i^*,j^*)$, the length of $P(v_{c^*},v_{a^*})$ is
strictly shorter than that of $P(v_{c^*},v_{j^*})$. Thus we
obtain that $d_{G({i^*},{j^*})}(v_{c^*},v_{j^*})=d_P(v_{c^*},v_{j^*})>r^*$. This incurs contradiction as $c^*$ is a
center of $G({i^*},{j^*})$ and the radius of
$G({i^*},{j^*})$ is $r^*$.

The claim implies that $P(v_{i^*},v_{c^*})\cup e(v_{i^*},v_{j^*})$ is a shortest path from $v_{c^*}$ to $v_{j^*}$ in $G({i^*},{j^*})$, and this further implies that $P(v_{i^*},v_{c^*})$ is a shortest path from $v_{c^*}$ to $v_{i^*}$ in $G({i^*},{j^*})$. Hence, for any $i\in [1,c^*]$, $P(v_{i},v_{c^*})$ is a shortest path from $v_{c^*}$ to $v_{i}$ in $G({i^*},{j^*})$ and $d_{G({i^*},{j^*})}(v_{c^*},v_i)=d_P(v_i,v_{c^*})$. Also, for any $i\in [j^*,n]$, $P(v_{i^*},v_{c^*})\cup e(v_{i^*},v_{j^*})\cup P(v_{j^*},v_i)$ is a shortest path from $v_{c^*}$ to $v_i$ in $G({i^*},{j^*})$ and $d_{G({i^*},{j^*})}(v_{c^*},v_i)=d_P(v_{i^*},v_{c^*})+|e(v_{i^*},v_{j^*})|+d_P(v_{j^*},v_i)$. We will use these properties below without further explanations.

Next we prove $d_{G({c^*},{j^*})}(v_{c^*},v_k)\leq r^*$.

If $k\in [1,c^*]$, then $d_{G({c^*},{j^*})}(v_{c^*},v_k)\leq d_P(v_k,v_{c^*})=d_{G({i^*},{j^*})}(v_{c^*},v_k)\leq r^*$.

If $k\in (c^*,j^*)$, then both $v_{c^*}$ and $v_k$ are in the cycle $C({i^*},{j^*})$ of the graph $G({i^*},{j^*})$ and are also in the cycle $C({c^*},{j^*})$ of the graph $G({c^*},{j^*})$. By a similar analysis as that for Lemma~\ref{lem:40}, $d_{G({c^*},{j^*})}(v_{c^*},v_k)\leq d_{G({i^*},{j^*})}(v_{c^*},v_k)\leq r^*$.

If $k\in [j^*,n]$, then by the triangle inequality, we have $d_{G({c^*},{j^*})}(v_{c^*},v_k)\leq |e(v_{c^*},v_{j^*})|+d_P(v_{j^*},v_k)\leq d_P(v_{i^*},v_{c^*})+ |e(v_{i^*},v_{j^*})|+d_P(v_{j^*},v_k)=d_{G({i^*},{j^*})}(v_{c^*},v_k)\leq r^*$.

The lemma thus follows.
\end{proof}


\paragraph{Summary.} We have computed a candidate solution for each of Case~1.1 and Case~1.2. Each candidate solution is also a feasible one. We have proved that if an optimal solution belongs to one of the two cases, then the corresponding candidate solution must also be an optimal solution. On the other hand, we have shown that other cases can be reduced to the two cases. Therefore, one of the two candidate solutions must be an optimal one. As a final step of our algorithm, among the two candidate solutions, we return the one with smaller candidate radius as our optimal solution. The running time of the entire algorithm is $O(n)$.

\begin{theorem}
The discrete-ROAP problem can be solved in linear time.
\end{theorem}








\section{Concluding Remarks}
\label{sec:conclude}

We presented a linear time algorithm for solving the discrete radius optimally augmenting path problem, which is optimal and matches the time complexity of the algorithm in~\cite{ref:JohnsonAL19} for the continuous version of the problem. As the algorithm in~\cite{ref:JohnsonAL19}, our algorithm uses configurations but the number of configurations in our algorithm is much fewer. From this point of view, it seems that the discrete problem is easier. However, this may not be the case due to the following.

Consider the problem of computing the radius for the resulting graph $G(i,j)$ after a new edge $e(i,j)$ is added to $P$ (e.g., see Fig.~\ref{fig:introduction}). For simplicity, let's consider an easier problem of computing the radius of the cycle $C(i,j)$. For the continuous problem, it is easy to see that the radius is equal to $|C(i,j)|$ minus the maximum length of all edges of $C(i,j)$. Therefore, once the maximum edge length of the subpath $P(v_i,v_j)$ is known (we then compare it with the length of $e(i,j)$), the radius of $C(i,j)$ can be computed in $O(1)$ time. As for computing the maximum edge length of $P(v_i,v_j)$, as shown in~\cite{ref:JohnsonAL19}, this can be trivially done in $O(\log n)$ time, after $O(n)$ time preprocessing on $P$. This is one reason that Johnson and Wang \cite{ref:JohnsonAL19} were able to provide an efficient algorithm for the {\em query version} of the problem: After $O(n)$ time preprocessing, given any query index pair $(i,j)$, the radius of the graph $G(i,j)$ can be computed in $O(\log n)$ time.

\begin{figure}[t]
    \centering
    \includegraphics[scale=0.55]{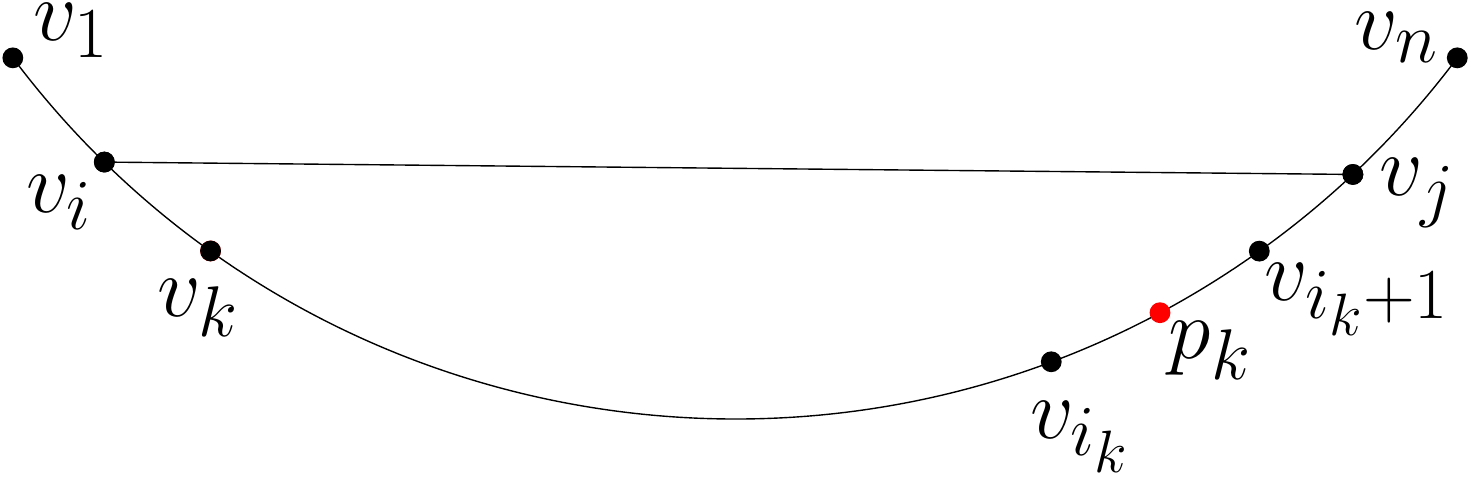}
    \caption{Illustrating the graph $G(i,j)$. $p_k$ is the diametral point of $v_k$ in the cycle $C(i,j)$.}
    \label{fig:introduction}
\end{figure}

Now consider the same problem (i.e., computing the radius of $C(i,j)$) for the discrete version. For each vertex $v_k\in C(i,j)$, let $p_k$ denote a conceptual point on an edge of $C(i,j)$ such that the length of the subpath of $P(v_k,p_k)$ between $v_k$ and $p_k$ is equal to $|C(i,j)|/2$, i.e., $p_k$ is the diametral point of $v_k$ in $C(i,j)$. Let $e(v_{i_k},v_{i_k+1})$ be the edge containing $p_k$. Define $l_k=\max\{d_{C(i,j)}(v_k,v_{i_k}),d_{C(i,j)}(v_k,v_{i_k+1})\}$. Then, it is easy to see that the radius of $C(i,j)$ is equal to $\min_{k\in [i,j]}l_k$. Although it is straightforward to compute the radius of $C(i,j)$ in $O(n)$ time, it is elusive to us whether it is possible to do so in $O(\log n)$ time (or just sub-linear time) with $O(n)$ time preprocessing. Unlike the continuous problem, this is an obstacle to solving the query problem in sub-linear time if only $O(n)$ (or slightly more) preprocessing time is allowed. We leave it as an open problem.



\footnotesize
\bibliographystyle{abbrv}
\bibliography{reference}




\end{document}